%% file: kruskalgoesgraphs.tex
\documentclass[submission,copyright,creativecommons, noncommercial]{eptcs}
\providecommand{\event}{TERMGRAPH 2016} 
\usepackage{breakurl}                   

\usepackage[utf8x]{inputenc} 
\usepackage[usenames,dvipsnames]{xcolor}
\usepackage{tikz}
\usetikzlibrary{fit}
\usetikzlibrary{positioning}
\usepackage[inline]{enumitem}

\usepackage{amsmath,amssymb,amsthm}
\theoremstyle{plain}
\newtheorem{corollary}{Corollary}
\newtheorem{lemma}[corollary]{Lemma}

\newtheorem{theorem}[corollary]{Theorem}
\theoremstyle{definition}
\newtheorem{definition}[corollary]{Definition}
\newtheorem{example}[corollary]{Example}

\DeclareMathAlphabet{\mathcal}{OMS}{cmsy}{m}{n}

\usepackage{pdfcomment}
\usepackage{tabularx}

\usepackage{wrapfig}

\title{Kruskal's Tree Theorem for Acyclic Term Graphs%
  \footnote{This work was partially supported by FWF (Austrian Science
    Fund) project P\ 25781-N15.}%
}
\author{Georg Moser
  \institute{Universität Innsbruck, Austria}
  \email{\href{mailto:georg.moser@uibk.ac.at}{georg.moser@uibk.ac.at}}
\and
Maria A. Schett
  \institute{Universität Innsbruck, Austria}
  \email{\href{mailto:maria.schett@uibk.ac.at}{maria.schett@uibk.ac.at}}
}

\def\titlerunning{Kruskal's Tree Theorem for Acyclic Term Graphs}
\def\authorrunning{G.~Moser \& M.A.~Schett}
\begin{document}

\hypersetup
{
    pdfauthor={Georg Moser, Maria Anna Schett},
    pdfsubject={Extended Abstract for TERMGRAPH 2016},
    pdftitle={Kruskal's Tree Theorem for Acyclic Term Graphs},
    pdfkeywords={Term Graph Rewriting, Termination, Kruskal's Tree Theorem}
}


\maketitle
\input{commands}
%
\begin{abstract}
  In this paper we study termination of term graph rewriting, where we
  restrict our attention to \emph{acyclic} term graphs. Motivated by
  earlier work by Plump we aim at a definition of the notion of
  \emph{simplification order} for acyclic term graphs. For this we
  adapt the homeomorphic embedding relation to term graphs. In
  contrast to earlier extensions, our notion is inspired by
  morphisms. Based on this, we establish a variant of Kruskal's Tree
  Theorem formulated for acyclic term graphs.  In proof, we rely on
  the new notion of embedding and follow Nash-Williams' minimal bad
  sequence argument. Finally, we propose a variant of the
  \emph{lexicographic path order} for acyclic term graphs.
\end{abstract}

\section{Introduction}

It is well-known that term graph rewriting is \emph{adequate} for term
rewriting.  However, this requires suitable care in the treatment of
sharing, typically achieved by extending the term graph rewrite
relation with \emph{sharing} (aka \emph{folding}) steps and
\emph{unsharing} (aka \emph{unfolding}) steps,
cf.~\cite{1998_plump,TeReSeC13}. If one focuses on term graph
rewriting alone, then it is well-known that termination of a given
graph rewrite system does not imply termination of the corresponding
term rewrite system~\cite{1997_plump}.  This follows as the
representation of a term as a graph enables us to share equal
subterms. However, if we do not provide the possibility to unshare
equal subterms, we change the potential rewrite steps.  Then not every
term rewrite step can be simulated by a graph rewrite step. This
motivates our interest in termination techniques \emph{directly} for
term graph rewriting.
More generally our motivation to study term graph rewriting stems from
ongoing work on complexity or termination analysis of programs based
on transformation to term rewrite systems (see
e.g.~\cite{GRSST:11,BMOG12,SGBFFHS14,ADM:2015,AM:16}). In
particular in work on termination of imperative programs
(see~e.g.~\cite{SGBFFHS14}) these works require a term representation
of the heap, which would be much more naturally be encoded as term
dags (see the definition below). However, complexity and termination
analysis of term graph rewrite systems have only recently be conceived
attention in the
literature~\cite{2013_bonfante_et_al,BKZ:14,2015_bruggink_et_al,AM:2016}. In
particular, at the moment there are no automated tools, which would
allow an application for program analysis and could be compared to
existing approaches using either \aprove~\cite{GBEFFOPSSST14} or
\tct~\cite{AMS:16}.

In our definition of term graph rewriting we essentially follow
Barendsen~\cite{TeReSeC13}, but
also~\cite{1987_barendregt_et_al,2013_avanzini}, which are
notationally closest to our presentation. 
We restrict our attention to term graphs, which represent such
(finite) terms, that is in our context term graphs are directed,
\emph{rooted}, and \emph{acyclic} graphs with node labels over a set
of function symbols and variables.
In term rewriting, termination is typically established via
compatibility with a reduction order. Well-foundedness of such an
order is more often than not a consequence of Kruskal's Tree
Theorem~\cite{1960_kruskal} (e.g. in~\cite{1982_dershowitz}). In
particular, Kruskal's Tree Theorem underlies the concept of simple
termination (see e.g.~\cite{1997_middeldorp_et_al}).  Indeed,
Plump~\cite{1997_plump} defines a simplification order for acyclic
term graphs. This order relies on the notion of \emph{tops}. The top
of a term graph is its root \emph{and} its direct successors---thus
keeping information on how these successors are shared. 

We recall briefly. Let $\wqo$ be a partial order.  If for any infinite
sequence, we can find two elements $a_i, a_j$ with $i < j$ where
$a_i \wqo a_j$, then $\wqo$ is a well-quasi order. Now, Kruskal's Tree
Theorem states, in a formulation suited to our needs, that given a
well-quasi order $\topembeq$ on the symbols in a term, the
\emph{homeomorphic embedding} relation $\embeq$ is a well-quasi order
$\embeq$ on terms.  We consider term graphs, not terms, and our
symbols are tops.  Usually, the relation $\embeq$ is simply called an
\emph{embedding}.

Plump~\cite{1997_plump} defines the embedding $\pembeq$, but as he
notes, for the following two term graphs, his definition of $\pembeq$
holds in both directions.
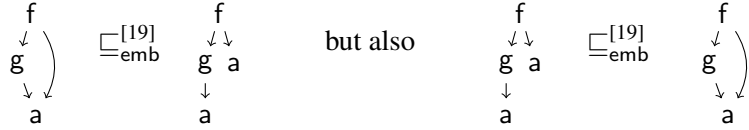
\begin{center}
  \pdftooltip{
    \begin{tikzpicture}
      \node [n]    (1) {$\f$};
      \node [n, xshift=0.5cm, %
             below left = of 1]  (2) {$\g$};
      \node [n, xshift=0.25cm,
             below = of 2]       (4) {$\fa$};
      \path[->] (1) edge (2)
                (1) edge [bend left] (4)
                (2) edge (4)
      ;
      \node [xshift=1.3cm, yshift=-0.4cm] (E) {$\pembeq$};
      \node [n,xshift = 2.5cm]   (1) {$\f$};
      \node [n, xshift=0.5cm, %
             below left = of 1]  (2) {$\g$};
      \node [n,
             below = of 2]       (4) {$\fa$};
      \node [n, xshift=-0.5cm, %
             below right = of 1] (3) {$\fa$};
      \path[->] (1) edge (2)
                (1) edge (3)
                (2) edge (4)
      ;
      \node [xshift=4.5cm, yshift=-0.4cm] (E) {
        \begin{tabular}{l}
          but also
        \end{tabular} 
      };
      \node [n, xshift = 6.5cm]  (1) {$\f$};
      \node [n, xshift=0.5cm, %
             below left = of 1]  (2) {$\g$};
      \node [n,
             below = of 2]       (4) {$\fa$};
      \node [n, xshift=-0.5cm, %
             below right = of 1] (3) {$\fa$};
      \path[->] (1) edge (2)
                (1) edge (3)
                (2) edge (4)
      ;
      \node [xshift=7.8cm, yshift=-0.4cm] (E) {$\pembeq$};
      \node [n, xshift = 9.2cm]    (1) {$\f$};
      \node [n, xshift=0.5cm, %
             below left = of 1]  (2) {$\g$};
      \node [n, xshift=0.25cm,
             below = of 2]       (4) {$\fa$};
      \path[->] (1) edge (2)
                (1) edge [bend left] (4)
                (2) edge (4)
      ;
    \end{tikzpicture}
  }
  { The term f(g(a),a) can be represented in two ways: once with the a
    shared, and once where it is not shared. They are, however,
    mutually embedded with respect to the embedding relation defined
    in \cite{1997_plump}. }
\end{center}
In particular, \cite{1997_plump} does not take sharing into
account---except for direct successors through tops. This is a
consequence of identifying each sub-graph independently.
This is the inspiration and starting point for our work: We want to
define an embedding relation, which also takes sharing into
account. With this new embedding relation we re-prove Kruskal's Tree
Theorem. Also here we take a slightly different approach to
\cite{1997_plump}, which relies on an encoding of tops to function
symbols with different arities. It is stated that there is a direct
proof based on~\cite{1963_nash-williams}, which will be our direction.

As already mentioned, the context of this paper is the quest for
termination techniques for term graph rewriting.  Here
\emph{termination} refers to the well-foundedness of the graph rewrite
relation $\grw$, induced by a graph rewrite system $\mathcal{G}$,
cf.~\cite{TeReSeC13}.  In particular, we seek a technique based on
orders.  This is in contrast to related work in the literature. There
termination is typically obtained through interpretations or weights,
cf.~Bonfante et al.~\cite{2013_bonfante_et_al}. Also Bruggink et
al.~\cite{BKZ:14,2015_bruggink_et_al} use an interpretation method,
where they use type graphs to assign weights to graphs to prove
termination. Finally, in~\cite{AM:2016} complexity of acyclic term
graph rewriting is investigated, based on the use of interpretations
and suitable adaptions of the dependency pair framework.

This paper is structured as follows. The next section provides basic
definitions.  In Section~\ref{Embedding} we discuss potential
adaptions of the homeomorphic embedding relation to term graphs and
establish a suitable notion that extends the notion of \emph{collapse}
known from the literature. Section~\ref{Kruskal} establishes our
generalisation of Kruskal's Tree Theorem to acyclic term graphs.  In
Section~\ref{Simplification:Orders} we establish a new notion of
simplification orders. Finally, in Section~\ref{Conclusion} we
conclude and mention future work.

\section{Preliminaries}
First, we introduce our flavour of term graphs based on term dags,
define term graph rewriting in our context, and give the collapse
relation. Then we investigate tops with respect to a function symbol
but also with respect to a node in a term graph. Based on this, we
will consider a precedence on tops.
\begin{definition} 
  Let $\ndes$ be a set of nodes, $\F$ a set of function symbols, and
  $\V$ a set of variables.
  A \emph{graph} is $G = (N, \succs, \lb)$, where
  $N \subseteq \ndes $, $\succs: N \to N^*$, and
  $\lb: N \to {\F \cup \V}$. Here, $\succs$ maps a node $n$ to an
  ordered list of \emph{successors} $[n_1 \ldots n_k]$. Further, $\lb$
  assigns labels, where
  \begin{enumerate*}[label=\itshape(\roman*)]
    \item for every node $n \in G$ with $\lb(n) = f \in \F$ we have
      $\succs(n) = [n_1, \dots , n_{\ar(f)}]$, and
    \item for every $n \in G$ with $\lb(n) \in \V$, we have
      $\succs(n) = [~]$.
  \end{enumerate*}  
  If $G$ is acyclic, then $G$ is a \emph{term dag}.
\end{definition}
The \emph{size} of a graph $|G|$ is the number of its nodes $N$. We
write $n \in G$ and mean $n \in N$, and call $G$ \emph{ground}, if
$\lb: N \to \F$.
If $\succs(n) = [\ldots, n_i, \ldots]$, we write $n \succi{i} n_i$, or
simply $n \succi{} n_i$ for any $i$.  Further,~$\succi{}^+$ is the
transitive, and $\succi{}^*$ the reflexive, transitive closure of $\succi{}$.
If $n \succi{}^* n\pr$, then $n\pr$ is \emph{reachable} from $n$. In
the sub-graph $G \sub [n_1, \ldots, n_k]$ all nodes reachable from
$n_1, \ldots, n_k$ are collected, i.e.\
$N = \{ n \mid n_i \succi{}^* n, 1 \leqslant i \leqslant k \}$, and
the domains of $\succs$ and $\lb$ are restricted accordingly.

\begin{definition}
  Let $T$ be a term dag.  If all nodes are reachable from one node
  called $\rt(T)$, that is, $T$ is \emph{rooted}, then $T$ is a
  \emph{term graph} with $\inlets \df \succs(\rt(T))$. For a term dag
  $G$ with $\inlets = [t_1, \ldots, t_n]$, the \emph{argument graph}
  is defined as $G \sub \inlets\pr$, where
  $\inlets\pr \df \succs(t_1) \cdots \succs(t_l)$.
\end{definition}

\noindent
\begin{minipage}[t]{.85\linewidth}
  \begin{example} \label{ex:termgraph} %
    On the right we show the term graph
    $T = (\{\cial{1}, \cial{2}\}, \succs, \lb)$, with
    $\succs: \cial{1} \mapsto [\cial{2}, \cial{2}], %
    \cial{2} \mapsto [~]$, and $\lb: \cial{1} \mapsto \f, %
    \cial{2} \mapsto \fa$.
    The term representation of $T$ is $\f(\fa,\fa)$, $|T| = 2$, and
    $T$ is ground. The argument graph of $T$ is ~ \nn{\fa}{2} ~ with
    $\inlets = [\cial{2}, \cial{2}]$.
  \end{example}
\end{minipage} \hfill
\begin{minipage}[t]{.1\linewidth} 
  \pdftooltip{
    \begin{tikzpicture}[anchor=base, baseline=10pt]
      \node [n]               (1) {\nn{\f}{1}};
      \node [n, below = of 1] (2) {\nn{\fa}{2}};
      \path[->] (1) edge [bend left]  (2)
                (1) edge [bend right] (2)
                ;
    \end{tikzpicture}
  }
  {A term graph representing the term f(a,a). Here the "a" is shared. The
   root node with label "f" has the node number 1. The argument node
   with label "a" has the node number 2. There are two edges from node 1
   to node 2.}
\end{minipage}
\newline

A \emph{graph rewrite rule} is a term dag $G$ with a root node~$l$
of the left hand side, and a root node~$r$ of right hand side. We
denote a graph rewrite rule by $L \to R$, where $G \sub [l] = L $ and $G \sub [r] =
R$. For a graph rewrite rule the following has to hold:
\begin{enumerate*}[label=\itshape(\roman*)]
  \item $\lb(l) \not\in \V$, 
  \item if $n \in R$ with $\lb(n) \in \V$ then $n \in L$, and
  \item for all nodes $n, n\pr \in G$, if $\lb(n) = \lb(n\pr) \in \V$
    then $n=n\pr$.
\end{enumerate*}
A \emph{graph rewrite system (GRS)} $\G$ is a set of graph rewrite
rules.
To define a \emph{graph rewrite step}, we first need the auxiliary
concepts of \emph{redirection} of edges and \emph{union} of two term dags.
%
To \emph{redirect} edges pointing from node $u$ to node $v$, we write
$G[v \rd u]$, which is defined as
$(N_G, \succs_{G[v \rd u]}, \lb_{G})$, where for all nodes $n \in G$,
$\succs^i_{G[v \rd u]}(n) \df v $ if $n = u$, and
$\succs^i_{G[v \rd u]}(n) \df n$ otherwise. Note, that for
$G[v \rd u]$ we still have $u \in G$.
%
For two term dags $G$ and $H$, their (left-biased) \emph{union},
denoted by $G \un H$, is defined as
$(N_G \cup N_H, \succs_{G} \un \succs_{H}, \lb_{G} \un \lb_{H}) $,
where for $f \in \{\succs, \lb\}$ we define
$ f_G \un f_H (n) \df f_G(n) $ if $n \in G$, and $f_H(n)$ if
$n \not\in G$ and $n \in H$.  Note, that we do not require
$N_G \cap N_H = \varnothing$.
Next we investigate how to determine whether a graph rewrite rule matches a
term graph. Therefore we first need to find a common structure between
two graphs---through a \emph{morphism}.
\begin{definition} \label{def:dmo}%
  Let $S, T$ be term graphs, and $\Delta \subseteq \F \cup \V$. A function $m: S \to T$ is
  \emph{morphic} if for a node $n \in S$
  \begin{enumerate}[label=\itshape(\roman*)]
  \item \label{dmo:lb} $\lb_S(n) = \lb_T(m(n))$ and
  \item \label{dmo:succs} if $n \succi{i}_{S} n_i $ then $m(n) \succi{i}_T m(n_i)$ for
    all appropriate $i$.
  \end{enumerate}
  A $\Delta$-\emph{morphism} from $S$ to $T$ is a mapping
  $m : S \to_\Delta T$, which is morphic in all nodes $n \in S$ with
  $\lb(n) \not\in \Delta$ and additionally $m(\rt(S)) = \rt(T)$ holds.
\end{definition}
A $\Delta$-morphism only enforces Conditions~\ref{dmo:lb} and
\ref{dmo:succs} on nodes with labels which are not in~$\Delta$. With
$\Delta = \V$ we can determine whether a left-hand side of a graph
rewrite rule \emph{matches} a term graph, i.e., $L$ matches $S$ if
there is a morphism $m : L \to_{\V} S$. Here, a node representing a
variable in~$L$ can be mapped 
to a node with any label and successors.
The morphism~$m$ is \emph{applied} to $R$, denoted by $m(R)$, by
redirecting all variable nodes in $R$ to their image. That is, for all
$n_1, \ldots, n_k \in R$, where $\lb(n_i) \in \V$, we define
$ m(R) = ((R \un S) [m(n_1) \rd v_1]) \ldots [m(n_k) \rd n_k]$.
Finally, for two term graphs $S, T$, $n$ a node in $S$, and
$N_S \cap N_T = \varnothing$, the replacement of the subgraph
$S\sub n$ by T, denoted $S[T]_n$, is defined as $T$, if $n = \rt(S)$,
and as $(S \un T)[\rt(T) \rd n] \sub \rt(S)$ otherwise.

\begin{definition}
  Let $\G$ be a GRS. A term graph $S$ \emph{rewrites} to a term graph
  $T$, denoted by $S \grw T$, if there is a graph
  rewrite rule $L \to R \in \G$ with $N_R \cap N_S = \varnothing$, and
  a morphism $m : L \to S\sub n$ such that $S[m(R)]_n = T$.
\end{definition}

Finally, we can introduce the notion of termination.
\begin{definition}
  If $\grw$ is well-founded, we say that the GRS $\G$ is
  \emph{terminating}.  
\end{definition}

So far, we have not taken sharing into
  account---which we will investigate next.
For term graphs $S$ and $T$, we may
ask: Is $S$ a ``more shared'' version of $T$? Are $S$ and $T$
``equal''? To answer this, 
we rely again on a morphism as in Definition~\ref{def:dmo}, where we
require Condition~\ref{dmo:lb} and~\ref{dmo:succs} for every node,
i.e.\ we set $\Delta = \varnothing$.
\begin{definition} \label{def:sharing}%
 If there is
  a morphism $m: S \to_\varnothing T$, then $S$ \emph{collapses}
  to $T$, denoted by
  $S \fleq T$.  If $S \fleq T$ and $T \fleq S$, then $S$ is
  \emph{isomorphic} to $T$, denoted by $S \iso T$.
\end{definition}
Reconsidering Example~\ref{ex:termgraph}, let $S$ be a tree
representation of $\f(\fa, \fa)$, then $S \fleq T$. Now recall, that
we aim to give a notion of $\Top$, which takes the sharing of
successor nodes into account, formalised via the collapse relation. 
Thus---with collapsing---we can give a
definition of $\Tops$ for a function symbol $f$.
\begin{definition} \label{def:tops} %
  Let $f \in \F$, $\ctg$ a fresh symbol wrt.\ $\F$, and $S$ a tree
  representation of $f(\ctg, \ldots, \ctg)$. Then
  $\Tops(f)= \{ T \mid T \text{ is a termgraph, and } S \fleq T \}$ and
  $\Tops(\F) = \bigcup_{f \in \F} \Tops(f)$.
\end{definition}
Now, similar to a precedence on function symbols, we define a
precedence $\topembeq$ on $\Tops(\F)$.
\begin{definition} \label{def:precedence} %
  A \emph{precedence} on $\F$ is a transitive relation $\topembeq$ on
  $\Tops(\F)$, where for $S, T \in \Tops(\F)$ we have
  \begin{enumerate*}[label=\itshape(\roman*)]
  \item \label{it:embTopRef} $S \iso T$ implies $S \topembeq T$ and
    $T \topembeq S$, and
  \item \label{it:embTopMo} $T \topembeq S$ implies $|T| \leqslant |S|$.
  \end{enumerate*}
\end{definition}
Condition~\ref{it:embTopRef} implies reflexivity, but also includes
isomorphic copies. Condition~\ref{it:embTopMo} hints at a major
distinction to the term rewriting setting: We can distinguish the same
function symbol with different degrees of sharing---and even embed
nodes, which are labelled with function symbols with a smaller arity,
in nodes labelled with function symbols with a larger arity. But, to
ensure that such an embedding is indeed possible, enough nodes have to
present---which is guaranteed by Condition~\ref{it:embTopMo}.
With Definition~\ref{def:tops} we can compute the $\Tops$ for a
function symbol---but we also want to compute the $\Top$ from some
node in a term dag.
\begin{definition} \label{def:top} %
  For a term dag $G = (N, \succs, \lb)$ and a node $n \in G$, we define
  $\Top_G(n) \df (\{n\} \cup \succs(n), \lb\pr, \succs\pr)$, where
  \begin{enumerate*}[label=\itshape(\roman*)]
  \item $\lb\pr(n) = \lb(n)$,
    $\succs\pr(n) = \succs(n)$, and
  \item for $n_i \in \succs(n)$, $ \lb\pr(n_i) = \ilab$, and
    $\succs\pr(n_i) = [~]$.
  \end{enumerate*} 
\end{definition}
For $\Top_G(n)$, where $\lb_G(n) = f$, we find an isomorphic copy $G\pr$
of $\Top_G(n)$ in $\Tops(f)$, i.e.
$\Top_G(n) \iso G\pr \in \Tops(f)$.

In the context of this work we focus on the graph rewrite relation
$\grw$ and not on a relation combined with any explicit collapsing relation
$\fleq$, as e.g., in~\cite{1997_plump}. In passing, we note that for
the below established notion of homeomorphic embedding a similar relation
to the collapse relation $\fleq$ is possible as in Plump's work, 
cf.~\cite[Lemma~24]{1997_plump}.

\section{On Embedding}
\label{Embedding}

Next we continually develop a suitable definition of \emph{homeomorphic embedding}
for term dags. To get an intuition for embedding of term graphs
consider the following example. 
\begin{figure}[t]
\begin{center}
  \pdftooltip{
    \begin{tikzpicture}
      \node [n]                  (1) {\nn{\f}{1}};
      \node [n, xshift=0.75cm, %
             below left = of 1]  (2) {\nn{\fa}{2}};
      \node [n, xshift=-0.75cm, %
             below right = of 1] (3) {\nn{\fa}{3}};
      \path[->] (1) edge (2)
                (1) edge (3)
      ;
      \node [xshift=1.5cm, yshift=-0.4cm] (E) {$\sembeq$};
      \node [n, xshift = 2.5cm] (4) {\nn{\g}{A}};
      \node [n, below = of 4] (5) {\nn{\fa}{B}};
      \path[->] (4) edge  (5)
      ;
      \node [xshift=3.5cm, yshift=-0.4cm] (E) {$\sembeq$};
      \node [n, xshift = 4.5cm] (4) {\nn{\f}{I}};
      \node [n, below = of 4] (5) {\nn{\fa}{II}};
      \path[->] (4) edge [bend left]  (5)
                (4) edge [bend right] (5)
      ;
      \node [xshift=6.25cm, yshift=-0.4cm] (E) {
        \begin{tabular}{l}
          with \\ precedence
        \end{tabular} 
      };
      \node [n, xshift=8cm]      (1) {$\f$};
      \node [n, xshift=0.5cm, %
             below left = of 1]  (2) {$\ilab$};
      \node [n, xshift=-0.5cm, %
             below right = of 1] (3) {$\ilab$};
      \path[->] (1) edge (2)
                (1) edge (3)
      ;
      \node [xshift=8.8cm, yshift=-0.4cm] (E) {$\stopembeq$};
      \node [n, xshift = 9.5cm] (4) {$\g$};
      \node [n, below = of 4] (5) {$\ilab$};
      \path[->] (4) edge (5);
      \node [xshift=10.2cm, yshift=-0.4cm] (E) {$\stopembeq$};
      \node [n, xshift = 11cm] (4) {$\f$};
      \node [n, below = of 4] (5) {$\ilab$};
      \path[->] (4) edge [bend left]  (5)
               (4) edge [bend right] (5)
      ;
    \end{tikzpicture}
  }
  { Shown are three term graphs and a precedence. The precedence is:
    the function symbol f with two distinct successors is larger than
    the function symbol g with one successor, which is larger than the
    function symbol f with two shared successors. This gives raise to
    the following embedding of term graphs: The term graph with root
    symbol f and two distinct successor nodes both labelled with
    function symbol a is larger than the term graph with root symbol g
    and one successor node labelled with the function symbol a. This
    term graph is larger than the term graph with root symbol f and a
    shared successor node also labelled with function symbol a.  
  }
\end{center}
\caption{Intuitive Embeddings}
\label{fig:ex:embedding}
\end{figure}
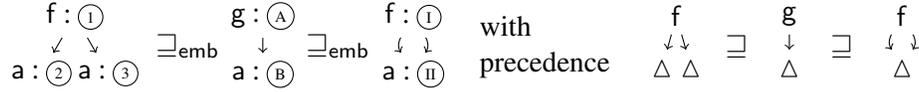

\begin{example} \label{ex:embedding} %
  In Figure~\ref{fig:ex:embedding}, we find three term graphs, which are
  intuitively embedded from left to right under the given precedence.
\end{example}

  We base our definition of embedding on morphisms. We evolve this
  definition to highlight difficulties and pitfalls. In the first
  attempt we try mapping nodes from the embedded to the embedding
  graph.
\begin{definition} [first attempt] 
\label{def:embeddingFstAttempt} 
  Let $\topembeq$ be a precedence. 
  We say that $S$ is \emph{embedded} in $T$, denoted as $S \embeq T$, if
  there exists a function $m \colon S \to T$, such
  that for all nodes $s \in S$, we have
  \begin{enumerate}[label=\itshape(\roman*)]
  \item $\Top_S(s) \topembeq \Top_T(m(s))$, and
  \item if $s \succi{}_S s\pr$ for some $s\pr \in S$, then
    $m(s) \succi{}_T^+ m(s\pr)$ holds.
  \end{enumerate}
\end{definition}

\begin{example}
  We illustrate this definition with Figure~\ref{fig:ex:embedding}. From the first to the second term
  dag we have a function $m$, with $m(\cial{A}) = \cial{1}$, and
  either $m(\cial{B}) = \cial{2}$ or $m(\cial{B}) = \cial{3}$. Here
  $m$ is not unique. From the second to the third term dag the
  morphism $m\pr$ maps $m\pr(\cial{I}) = \cial{A}$ and
  $m\pr(\cial{II}) = \cial{B}$.
\end{example}

In Definition~\ref{def:embeddingFstAttempt} the morphism maps nodes
from the embedded graph $S$ to nodes in the embedding graph $T$. The
following example shows a problem arising from this.

\begin{example} \label{ex:StoT} %
    The embedding given in Figure~\ref{fig:variants}(a) is valid after
    Definition~\ref{def:embeddingFstAttempt}. Here a morphism that
    satisfies both conditions is $\mm{A}{1}$, $\mm{B}{2}$,
    $\mm{C}{3}$, and also $\mm{D}{2}$ as well as $\mm{E}{3}$. This
    embedding could be prohibited by demanding $m$ to be injective.
\end{example}

Demanding injectivity in Definition~\ref{def:embeddingFstAttempt}
prohibits the embedding $S \embeq T$ if $S \ufleq T$ (in general).
Thus we attempt to expand our definition such that a term dag also
embeds a collapsed version of itself, i.e.\ embedding takes sharing
into account. To achieve this the embedding relation has to contain
the collapse relation of Definition~\ref{def:sharing}.  Then the
embedding relation relies on a \emph{partial} mapping from the
embedding term graph $S$ to $T$.
 
\begin{definition}[second attempt] \label{def:embeddingSndAttempt} %
  Let $\topembeq$ be a precedence.  We say that $S$ \emph{embeds} $T$,
  denoted as $S \sembeq T$, if there exists a partial, surjective
  function $m \colon S \to T$, such that for all nodes $s$ in the
  domain of~$m$, holds
  \newcounter{embenum}
  \begin{enumerate}[label=\itshape(\roman*)]
  \item
    $\Top_T(m(s)) \topembeq \Top_S(s)$, and
  \item %
    $m(s) \succi{}_T m(s\pr)$ implies 
    $ s \succi{}^{+}_S n\pr$ for some $ n\pr \in \{ n \mid (n) = m(s\pr) \}$. 
  \setcounter{embenum}{\value{enumi}}
  \end{enumerate}
\end{definition}
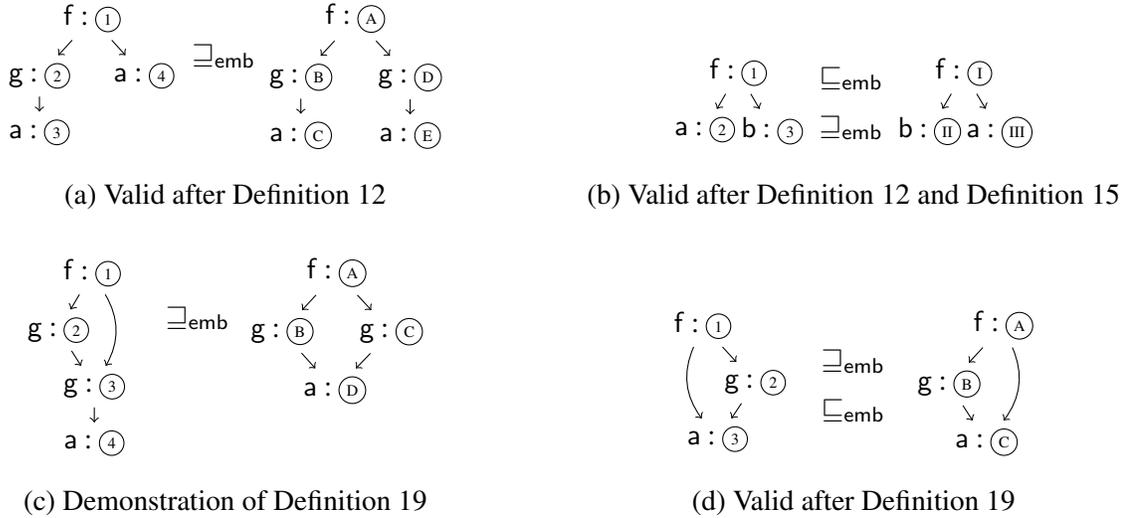
\begin{figure}[t]
  \centering
  \begin{tabular}{@{}c@{\hspace{.02\linewidth}}c@{}}
    \pdftooltip{ %
    \begin{minipage}[b]{.49\linewidth}
      \begin{center}
      \begin{tikzpicture}
        \node [n]                     (1) {\nn{\f}{1}};
        \node [n, xshift=0.5cm, %
               below left= of 1]     (2) {\nn{\g}{2}};
        \node [n, below       = of 2] (3) {\nn{\fa}{3}};
        \node [n, xshift=-0.5cm,
               below right = of 1]    (4) {\nn{\fa}{4}};
        \path[->] (1) edge (2)
                  (1) edge (4)
                  (2) edge (3)
                  ;
        \node [xshift=1.75cm, yshift=-0.5cm] (E) {$\sembeq$};
        \node [n, xshift= 3.5cm]      (1) {\nn{\f}{A}};
        \node [n, xshift=0.5cm, 
               below left  = of 1]    (2) {\nn{\g}{B}};
        \node [n, below       = of 2] (4) {\nn{\fa}{C}};
        \node [n, xshift=-0.5cm, 
               below right = of 1]    (3) {\nn{\g}{D}};
        \node [n, below       = of 3] (5) {\nn{\fa}{E}};
        \path[->] (1) edge (2)
                  (1) edge (3)
                  (2) edge (4)
                  (3) edge (5)
                  ;
      \end{tikzpicture} \\[1ex] %
      (a)~Valid after Definition~\ref{def:embeddingFstAttempt} %
      \end{center} %
     \end{minipage}
     }{ 
      Shown is the tree representation of $f:1 (g:2(a:3),a:4)$ and
      $f:A(g:B(a:C),g:D(a:E))$, where the former embeds the latter. 
      This is valid after the first attempt of the definition. %
     }
 & 
   \pdftooltip{
   \begin{minipage}[b]{.49\linewidth}
     \begin{center}
       \begin{tikzpicture}
         \node [n]                  (1) {\nn{\f}{1}};
         \node [n, xshift=0.75cm, %
         below left = of 1]  (2) {\nn{\fa}{2}};
         \node [n, xshift=-0.75cm, %
         below right = of 1] (3) {\nn{\fb}{3}};
         \path[->] (1) edge (2)
                   (1) edge (3)
                   ;
         \node [xshift=1.5cm, yshift=-0.4cm] (E) {
           \begin{tabular}{l}
             $\embeq$ \\[1ex] $\sembeq$
           \end{tabular} 
         };
         \node [n, xshift= 3cm]    (1) {\nn{\f}{I}};
         \node [n, xshift=0.75cm, %
         below left = of 1]  (2) {\nn{\fb}{II}};
         \node [n, xshift=-0.75cm, %
         below right = of 1] (3) {\nn{\fa}{III}};
         \path[->] (1) edge (2)
         (1) edge (3)
         ;
       \end{tikzpicture} \\[1ex]
       (b) Valid after Definition~\ref{def:embeddingFstAttempt}
         and Definition~\ref{def:embeddingSndAttempt}
     \end{center} 
   \end{minipage}
    }{ The term graph representing f(a,b) is embedded in the term graph
    representing f(b,a) and vice versa.}
\\[3ex]
  \pdftooltip{ %
  \begin{minipage}[b]{.49\linewidth}
    \begin{center}
      \begin{tikzpicture}
        \node [n]                  (1) {\nn{\f}{1}};
        \node [n, xshift=0.75cm, %
               below left = of 1]  (2) {\nn{\g}{2}};
        \node [n, xshift=-0.75cm, %
               below right = of 2] (3) {\nn{\g}{3}};
        \node [n,
               below = of 3]       (4) {\nn{\fa}{4}};
        \path[->] (1) edge (2)
                  (1) edge [bend left]  (3)
                  (2) edge (3)
                  (3) edge (4)
                  ;
        \node [xshift=1.4cm, yshift=-0.6cm] (E) {$\sembeq$};
        \node [n, xshift= 3.25cm]    (1) {\nn{\f}{A}};
        \node [n, xshift=0.5cm, %
               below left = of 1]  (2) {\nn{\g}{B}};
        \node [n, xshift=-0.5cm, %
               below right = of 1] (3) {\nn{\g}{C}};
        \node [n, xshift=-0.75cm,
               below = of 3] (4) {\nn{\fa}{D}};
        \path[->] (1) edge (2)
                  (1) edge (3)
                  (2) edge (4)
                  (3) edge (4)
                  ;
      \end{tikzpicture}
      \\[1ex]
      (c) Demonstration of  Definition~\ref{def:embeddingFinal}   
    \end{center}
  \end{minipage}
  }{ 
    The term graph f:1(g:2(x),x) where x= g:3(a:4) embeds the term graph
    f:A(g:B(y),g:C(y)) where y = a:D.}
&
  \pdftooltip{
  \begin{minipage}[b]{.49\linewidth}
    \begin{center}
      \begin{tikzpicture}
        \node [n]                  (1) {\nn{\f}{1}};
        \node [n, xshift=-0.5cm, %
               below right = of 1] (2) {\nn{\g}{2}};
        \node [n, xshift=-0.5cm,
               below = of 2]       (3) {\nn{\fa}{3}};
        \path[->] (1) edge (2) 
                  (1) edge [bend right] (3)
                  (2) edge (3)
                  ;
        \node [xshift=2cm, yshift=-0.8cm] (E) {
          \begin{tabular}{l}
            $\sembeq$ \\[1ex] $\embeq$
          \end{tabular} 
        };
        \node [n,xshift = 4cm]      (A) {\nn{\f}{A}};
        \node [n, xshift=0.5cm, %
               below left = of A]   (B) {\nn{\g}{B}};
        \node [n, xshift=0.5cm, %
               below = of B]        (C) {\nn{\fa}{C}};
        \path[->] (A) edge (B)
                  (B) edge (C)
                  (A) edge [bend left] (C)
                  ;
      \end{tikzpicture} \\[1ex]
      (d) Valid after Definition~\ref{def:embeddingFinal}
    \end{center}
  \end{minipage}
  }{ 
    Two term graphs are mutually embedded. The first one is f:1(x,g:2(x))
    where x = a:3, and the other one is f:A(g:B(y),y), where y = a:C.
  }%
\end{tabular}
\caption{Variants of Embedding}
\label{fig:variants}
\end{figure}

  \begin{example}
   Again consider Figure~\ref{fig:ex:embedding}. From the first to
  the second term dag we have a function $m$, with
   $m(\cial{1}) = \cial{A}$, $m(\cial{2}) = \cial{B}$, and/or
   $m(\cial{3}) = \cial{B}$. Here $m$ is not unique. From the second to
   the third term dag the morphism $m\pr$ maps
   $m\pr(\cial{A}) = \cial{I}$ and $m\pr(\cial{B}) = \cial{II}$.
  \end{example}

  One may observe, that both definitions of embedding so far are very
  permissive: it does not regard the order of the arguments. This is
  best illustrated by an example. 

\begin{example} \label{ex:fabfba}
  The two term graph shown in Figure~\ref{fig:variants}(b) 
  representing the terms $\f(\fa,\fb)$ and
  $\f(\fb,\fa)$ are mutually embedded: %
    from left to right we have the morphism $m$ with
    $m(\cial{1}) = \cial{I}$, $m(\cial{2}) = \cial{III}$, and
    $m(\cial{3}) = \cial{II}$. But---the inverse morphism $m^{-1}$
    also fulfils both conditions in
    Definition~\ref{def:embeddingFstAttempt} and
    Definition~\ref{def:embeddingSndAttempt}.
\end{example}

To remedy this, we need to take the order of the arguments into
account. Informally speaking, we want the preserve the relative order
between the nodes: if a node~$n$ is ``left of'' a node~$n\pr$, $m(n)$
should be ``left of'' $m(n\pr)$ in the embedded graph. For a formal
description of ``left of'', we employ \emph{positions}. Positions are
sequences of natural numbers with $\cdot$ as delimiter. The set of
positions of a node~$n$ in a term graph~$S$ is defined as follows:
$\Pos_S(n) \df \{\epsilon \}$ if $n = \rt(S)$, 
and
$\Pos_S(n) \df \{ p \cdot i \mid \exists n\pr \in S \text{ with } n\pr
\succi{i}_S n \text{ and } p \in \Pos_S(n\pr) \}$ otherwise.
For a term dag $G$ with $\inlets_G$, the base case is adapted
slightly: $\Pos_G(n) \df \{ i \} $ if $n$ is on $i$th position in
$\inlets_G$.
We can now compare two positions $p$ and $q$: $p$ is left---or
above---of $q$, if $p = p_1 \cdots p_k \lex q_1 \cdots q_l = q$, i.e.\
$p_i = q_i$ for $1 \leqslant i \leqslant j$ and $j = k < l$ or
$p_j < q_j$.

We now have to extend this comparison from positions to nodes. This
entails on the one hand an intra-node comparison which finds the
smallest position within a node. Then an inter-node comparison
comparing the smallest positions of two nodes.  Two nodes are called
\emph{parallel} in a term graph $G$, if they are mutually unreachable.

\begin{definition}
  Let $G$ be a term dag. We define a partial order $\lof_G$ on the
  parallel nodes in $G$. Let $n, n\pr \in G$ and suppose $n$ and
  $n\pr$ are parallel.  Further, suppose $p \in \Pos(n)$ is minimal
  wrt.\ $\lex$ and $q \in \Pos(n\pr)$ is minimal wrt.\ $\lex$.  Then
  $n \lof_G n\pr$ if $p \lex q$.
\end{definition}

Based on the above definition, we develop
Definition~\ref{def:embeddingSndAttempt} further to the final version
of embedding.

\begin{definition}[final] \label{def:embeddingFinal}
  Let $\topembeq$ be a precedence. 
  We say that $S$ \emph{embeds} $T$, denoted as $S \sembeq T$, if
  there exists a partial, surjective function $m \colon S \to T$, such
  that for all nodes $s$ in the domain of~$m$, holds
  \begin{enumerate}[label=\itshape(\roman*)]
  \item \label{embed:1}
    $\Top_T(m(s)) \topembeq \Top_S(s)$, and
  \item \label{embed:2} $m(s) \succi{}_T m(s\pr) $ implies 
    $ s \succi{}^{+}_S n\pr$ for some $ n\pr \in \{ n \mid m(n) = m(s\pr) \}$, and
  \item \label{embed:3} $m(s) \lof_T m(s\pr)$ implies either
    that none of the nodes in the preimage of $m(s\pr)$ is parallel to $s$, or
    there exists $n\pr \in \{ n \mid m(n) = m(s\pr) \}$ such that 
    $s \lof_S n\pr$.
   \end{enumerate}
\end{definition}

\begin{example}
  Recall Example~\ref{ex:fabfba}. With the final definition of
  embedding, the two term graphs are not mutually embedded per
  se---embedding now depends on $\topembeq$. 
  As a further example for the embedding of the term graphs
  consider Figure~\ref{fig:variants}(c). We have the following
  morphism: $\mm{1}{A}$, $\mm{2}{B}$, $\mm{3}{C}$, and
  $\mm{4}{D}$. Here we have $\cial{B} \lof \cial{C}$ but $\cial{2}$
  and $\cial{3}$ are not parallel.
  However, even with $\lof$ the two graphs in
  Figure~\ref{fig:variants}(d) are mutually embedded. Here we have
  neither $\cial{2} \lof \cial{3}$ nor $\cial{B} \lof \cial{C}$, so
  Condition~\ref{embed:3} holds trivially in both directions.
\end{example}

The relation $\sembeq$ is transitive, i.e.\ $S \sembeq T$ and
$T \sembeq U$ implies $S \sembeq U$. The proof is straight forward: We
construct the embedding $m_3 : S \to U$, based on the implied
embeddings $m_2 : S \to T$ and $m_1: T \to U$, by setting
$m_3(n) = m_1(m_2(n))$ and show that $m_3$ fulfils the conditions in
Definition~\ref{def:embeddingFinal}. 

\section{Kruskal's Tree Theorem for Acyclic Term Graphs}
\label{Kruskal}
Our proof follows \cite{1997_middeldorp_et_al} for the term rewrite
setting, which in turn follows the minimal bad sequence argument of
Nash-Williams \cite{1963_nash-williams}: we assume a minimal ``bad''
infinite sequence of term graphs and construct an even smaller ``bad''
infinite sequence of their arguments. By minimality we contradict that
this sequence of arguments is ``bad'', and conclude that it is
``good''. So we start by defining the notions of ``good'' and ``bad''.

\begin{definition}
  Assume a reflexive and transitive order $\wqo$, and an infinite
  sequence~$\asq$ with $a_i, a_j$ in $\asq$. If for some $i < j$ we
  have $a_i \wqo a_j$, then $\asq$ is \emph{good}. Otherwise,
  $\asq$ is \emph{bad}. If every infinite sequence is good, then
  $\wqo$ is a \emph{well-quasi order} (wqo).
\end{definition}
After we determined the sequence of arguments to be good, we want
to--- roughly speaking---plug the $\Top$ back on its argument. For
this, we need a wqo on $\Tops(\F)$ and the following, well
established, lemma.
\begin{lemma} \label{Lem:Chain} %
  If $\wqo$ is a wqo then every infinite sequence contains a
  subsequence---a \emph{chain}---with $a_i \wqo a_{i+1}$ for all $i$.
\end{lemma}
With this lemma, we can construct witnesses that our original minimal
bad sequence of term graphs is good, contradicting its badness and
concluding the following theorem.
\begin{theorem} \label{t:kruskal} %
  If $\topembeq$ is a wqo on $\Tops(\F)$, then $\embeq$ is a wqo on
  ground, acyclic term graphs.
\end{theorem}
\begin{proof}
  By definition, $\embeq$ is a wqo, if every infinite sequence is
  good, i.e.\ for every infinite sequence of term graphs, there are
  two term graphs $T_i, T_j$, such that $T_i \embeq T_j$ with
  $1 \leqslant i < j$.
  We construct a minimal bad sequence of term graphs $\mbsq$: Assume
  we picked $T_1, \ldots, T_{n-1}$. We next pick $T_n$---minimal with
  respect to $|T_n|$---such that there are bad sequences that start
  with $T_1, \ldots, T_n$.

  Let $G_i$ be the argument graph of the $i$th term graph $T_i$. We
  collect in $G$ the arguments of all term graphs of $\mbsq$, i.e.
  $G = \bigcup_{i \geqslant 1} G_i$ and show that $\embeq$ is a wqo on
  $G$.
  For a contradiction, we assume $G$ admits a bad sequence
  $\argsq$. We pick $G_k \in G$ with $k \geqslant 1$ such that
  $H_1 = G_k$. In $G\pr$ we collect all argument graphs up to $G_k$,
  i.e.  $G' = \bigcup_{i \geqslant 1}^k G_i$. The set $G\pr$ is finite,
  hence there exists an index $l > 1$, such that for all $H_i$ with
  $i \geqslant l$ we have that $H_i \in G$ but $H_i \not\in G\pr$.  We
  write $\argsq_{\geqslant l}$ for the sequence $\argsq$ starting at
  index~$l$.
  Now consider the sequence
  $T_1, \ldots, T_{k-1}, G_k, \argsq_{\geqslant l}$. By minimality of
  $\mbsq$ this is a good sequence. So we try to find a witness and
  distinguish on $i, j$:
  
  \noindent
  \begin{tabularx}{\textwidth}{l X}
  $\underbrace{ T_1, \ldots, T_{k-1}}_{i,j}, G_k, \argsq_{\geqslant l}$ &%
   For $1 \leqslant i < j \leqslant k-1$, we have
   $T_i \embeq T_j$, which contradicts the badness of~$\mbsq$.\\
   $\underbrace{ T_1, \ldots, T_{k-1}}_{i}, \underbrace{G_k}_{j},
   \argsq_{\geqslant l} $ &%
   For $1 \leqslant i \leqslant k-1$ and $j = k$, we have
   $T_i \embeq G_k$ and $G_k \embeq T_k$, where the latter is a direct
   consequence of the definitions. Hence, by transitivity,
   $T_i \embeq T_j$, which contradicts the badness of
   $\mbsq$. \\
   $\underbrace{ T_1, \ldots, T_{k-1}}_{i}, G_k,
   \underbrace{\argsq_{\geqslant l}}_{j}$ &%
   For $1 \leqslant i \leqslant k-1$ and $j \geqslant l$, we have
   $H_j \not\in G\pr$ by construction, but then $H_j = G_m$ for some $m > k$ 
   and thus $H_j \embeq T_m$. Together with $T_i \embeq H_j$, we obtain
   $T_i \embeq T_m$ by transitivity, which contradicts the
   badness of $\mbsq$. \\
  %
   $T_1, \ldots, T_{k-1},  \underbrace{G_k, \argsq_{\geqslant l}}_{i,j}$ &%
   Hence for some $1 \leqslant i < j$, where
   $i,j \not\in \{2, \ldots, l-1\}$, we have some $H_i \embeq H_j$, which
   contradicts the badness of $\argsq$.
  \end{tabularx}
  \noindent
  We conclude $\argsq$ is a good sequence and $\embeq$ is wqo on $G$.
  
  Next we consider the $\Top$s of $\mbsq$. Let these $\Top$s be
  $\topf$. By assumption, $\topembeq$ is a wqo on $\Tops(\F)$, and by
  Lemma~\ref{Lem:Chain}, $\topf$ contains a chain $\topf_\phi$, i.e.
  $f_{\phi(i)} \topembeq f_{\phi(i+1)}$ for all $i \geqslant 1$. We
  proved $\embeq$ to be a wqo on $G$. Hence we have
  $G_{\phi(i)} \embeq G_{\phi(j)}$ for some $1 \leqslant i < j$.
  It remains to be shown, that $f_{\phi(i)} \topembeq f_{\phi(j)}$ and
  $G_{\phi(i)} \embeq G_{\phi(j)}$ implies
  $T_{\phi(i)} \embeq T_{\phi(j)}$.
  We construct $T_{\phi(i)}$, and analogous $T_{\phi(j)}$, from
  $f_{\phi(i)} = (n_i, \lb_{f\phi(i)}, \succs_{f\phi(i)})$ and
  $G_{\phi(i)} = (N_{G\phi(i)}, \lb_{G\phi(i)}, \succs_{G\phi(i)})$
  with $\inlets_{G\phi(i)}$.  We have
  ${{N_{G\phi(i)}} \cap \{n_i\}} = \varnothing $.  Then
  $T_{\phi(i)} = (N_{T\phi(i)}, \lb_{T\phi(i)}, \succs_{T\phi(i)})$
  where
  \begin{enumerate}[label=\itshape(\roman*)]
  \item  the nodes $N_{T\phi(i)} \df {N_{G\phi(i)}} \cup \{n_i \}$,
  \item $\lb_{T\phi(i)} \df \lb_{G\phi(i)}$ extended by
    $\lb_{T\phi(i)}(n_i) = \lb_{f\phi(i)}(n_i)$, and
  \item $\succs_{T\phi(i)} \df \succs_{G\phi(i)}$ extended by
    $\succs_{T\phi(i)} (n_i) = \inlets_{G\phi(i)}$.
  \end{enumerate}
  
  We aim for $T_{\phi(i)} \embeq T_{\phi(j)}$ and therefore construct
  the morphism $m: T_{\phi(j)} \to T_{\phi(i)}$.  From
  $G_{\phi(i)} \embeq G_{\phi(j)}$, we obtain a morphism
  $m_G : G_{\phi(j)} \to G_{\phi(i)}$. We set $m(n) = m_G(n)$ for
  $n \in G_{\phi(j)}$, and $m(n_j) = n_i$.
  It remains to be shown that $m$ fulfils
  Definition~\ref{def:embeddingFinal}. Surjectivity of $m$ follows directly
  from the surjectivity of $m_G$. Condition~\ref{embed:1} holds for
  all nodes in $m_G$, and by $f_{\phi(i)} \topembeq f_{\phi(j)}$ also
  for $\rt(T_{\phi(j)}) = n_j$. For Condition~\ref{embed:2} we have to
  show:
  If $m(n_j) \succi{}_{T_{\phi{i}}} n\pr_i = m(n\pr_j)$ then
  $n_j \succi{}^{+} n\pr_j$.
  By definition $n\pr_i \in \inlets_{G\phi(i)}$ and hence also
  $n\pr_i \in G_{\phi(i)}$. By surjectivity of $m_G$ exist
  $m_G(n\pr_j) = n\pr_i$. It remains to be shown that
  $n_j \succi{}^+ n\pr_j$.
  By definition $n_j \succi{} u_j$, where
  $u_j \in \inlets_{G\phi(j)}$. By definition of argument graph, all
  nodes in $G_{\phi(j)}$ are reachable from nodes in
  $\inlets_{G\phi(j)}$, and in particular
  $n_j \succi{} u_j \succi{}^* n\pr_j$. 
  Finally, Condition~\ref{embed:3} holds trivially for $n$ and by
  $G_{\phi(i)} \embeq G_{\phi(j)}$.
  Hence we found a $T_{\phi(i)} \embeq T_{\phi(j)}$, which contradicts
  the badness of $\mbsq$. Therefore $\mbsq$ is good and~$\embeq$ is a
  wqo.
\end{proof}

\section{Simplification Orders}
\label{Simplification:Orders}

In the term rewriting setting simplification orders are defined
through the embedding relation. That is, a rewrite order~$\simpo$ is a
\emph{simplification order} if ${\emb} \subseteq {\simpo}$
\cite{1997_middeldorp_et_al}. Then, if we can orient the rules in a
rewrite system with $\simpo$, there are no infinite rewrite sequences.
We try to directly transfer this idea to the term graph rewriting
setting---but this is not sufficient, as the following example shows.
\begin{example}
  We can orient the rule on the left with $\semb$, but still may get
  an infinite rewrite sequence, as shown on the right.
\begin{center}
 \pdftooltip{
  \begin{tikzpicture}
     \femb{\fa}{\fa}{1}
    \node [xshift=2cm, yshift=-0.4cm] (E1) {$\semb$};
    \fembs{\fa}{3}
    \femb{\fa}{\fa}{6}
    \node [xshift=7cm, yshift=-0.4cm] (E) {$\grw$};
    \fembs{\fa}{8}
    \node [xshift=9cm, yshift=-0.4cm] (E2) {$\grw$};
    \fembs{\fa}{10}
    \node [xshift=11cm, yshift=-0.4cm] (E2) {$\ldots$};
  \end{tikzpicture}
}{ On the left we have the rule: first, on the left hand side is the
  tree representation of the term f(a,a), second, on the right hand
  side, the subterm a is shared. Here the left hand side strictly
  embeds the right hand side. So on the right we have an infinite
  rewrite sequence. The rewrite sequence starts with the tree
  representation of f(a,a), and performs one step to share the term a.
  But then the rule is applicable again, and again, ...
}
\end{center}
Note, that this infinite rewrite sequence is not bad wrt.\
$\embeq$.
\end{example}

This problem is \emph{not} caused by our definition of embedding, and
also occurs in \cite{1997_plump}. Rather, the reason is that from
orientation of the rules, we cannot conclude orientation of all
rewrite steps.
However, it should be noted, that the definition of simplification
order in~\cite{1997_plump} is indeed transferable to our presentation.

\begin{definition}[\cite{1997_plump}]
\label{d:simpo}
  Let $\embeq$ be the embedding relation induced by a precedence $\topembeq$ 
  that is a wqo. A transitive relation $\simpo$ is a \emph{simplification
    order}, if
  \begin{enumerate} [label=\itshape(\roman*)]
  \item $\emb \subset \simpo$, and
  \item for all $S$ and $T$, if $S \embeq T$ and $T \embeq S$ then
    $S \not\simpo T$.
  \end{enumerate}
\end{definition}

A direct consequence of the second condition is that simplification
orders are irreflexive. We obtain the following theorem.
\begin{theorem}
Every simplification order is well-founded.  
\end{theorem}
\begin{proof}
  Let $\succ$ denote a simplification order.  Thus there exists a
  well-quasi ordered precedence and an induced embedding relation,
  such that its strict part $\emb$ is contained in $\succ$. Due to
  Theorem~\ref{t:kruskal}, $\emb$ is a well-quasi order. Further, by
  definition $\succ$ is an irreflexive and transitive extension of
  $\embeq$. Thus $\succ$ is well-founded.
\end{proof}

Based on this observation, we adapt the definition of a
\emph{lexicographic path order} (\emph{LPO} for short) from term
rewriting to term graph rewriting and thus have a technique to show
  termination directly for acyclic term graph rewriting. Based on the above definition
of embedding, it is natural to define LPO on term dags. Thus, we
obtain the following definition of $\lpo$ induced by a well-quasi
ordered precedence.

\begin{definition}
Let $\topembeq$ be a well-quasi ordered precedence. We
write $\toplex$ for the lexicographic extension of $\sqsubset$.
  Let $S, T$ be term dags with $\inlets_S = [s_1, \ldots, s_k]$ and
  $\inlets_T = [t_1, \ldots, t_k]$, where $s_i$, $s_j$ and $t_i$,
  $t_j$ are parallel. Then $T \lpo S$ if one of the following holds
  \begin{enumerate}[label=\itshape(\roman*)]
  \item \label{lpo:1} %
    $T \lpoeq S \sub [s_{i_1}, \ldots, s_{i_{k'}}]$ for some
    $1 \leqslant i_1 < \ldots < i_{k'} \leqslant k$, or
  \item \label{lpo:2} %
    $[\Top(t_1), \ldots, \Top(t_l)] \toplex [\Top(s_1), \ldots,
    \Top(s_k)]$ and $\args(T) \lpo S$, or
  \item \label{lpo:3} %
    $[\Top(t_1), \ldots, \Top(t_l)] = [\Top(s_1), \ldots,
    \Top(s_k)]$ and $\args(T) \lpo \args(S) $.
  \end{enumerate}
\end{definition}

\begin{example}
  Recall Example~\ref{ex:fabfba}. Given the precedence
  $\fa \topemb \fb$ we can orient the two term graphs: from right to
  left. To orient the term graphs wit $\lpo$ we first use \ref{lpo:3}
  and compare the argument graphs. Then we compare their respective
  $\inlets$ lexicographically, i.e.,
  $[\Top(\cial{2}), \Top(\cial{3})] \toplex [ \Top(\cial{II}),
  \Top(\cial{III})]$ using \ref{lpo:2}.
\end{example}
To prove that $\lpo$ contains $\embeq$ for term graphs, it important
to note that $\lpo$ requires that nodes are parallel within
$\inlets$. That means, we can inductively step through a term graph,
with $\inlets$ forming a level in the term graph. With~\ref{lpo:1} we
can project the largest term dag to the dag that is actually used in
the embedding.

\section{Conclusion and Discussion}
\label{Conclusion}

Inspired by~\cite{1997_plump} we defined an embedding relation for the
term graph rewriting flavour of \cite{2013_avanzini,AM:2016} and 
re-proved Kruskal's Tree Theorem. Furthermore, based on Plump's work~\cite{1997_plump}
we establish a new notion of simplification order for acyclic term graphs
and provide a suitable adaption of the lexicographic path order to acyclic
term graphs.

In contrast to~\cite{1997_plump}, where the proof
 uses an encoding of $\Top$ to
 function symbols with different arities, 
 our proof operates on term graphs. 
With a new definition of the embedding relation, based on the
notion of morphism and taking sharing into account, and a new definition
of arguments we finally showed Kruskal's Tree Theorem for term graphs:
A well-quasi order on $\Tops$, i.e.\ $\topembeq$, induces a well-quasi
order~$\embeq$ on ground term graphs.
One insight from our proof concerns the arguments of a
term graph---or rather \emph{the} argument. For a term structure we
have several subterms as arguments. For a term graph structure it is
beneficial to regard the arguments as only one single argument
graph. This preserves sharing. Moreover a single argument simplifies
the proof as extending the order to sequences, Higman's
Lemma~\cite{1952_higman}, can be omitted. 

In future work, we will focus on the establishment of genuinely novel
notions of simplification orders for term graph rewriting and investigate
suitable adaptions of reduction orders for complexity analysis. 

\bibliographystyle{eptcs}
\bibliography{biblio}
\end{document}

%% file: commands.tex

\newcommand*{\cmt}[1]{\textcolor{White}{\colorbox{PineGreen}{\texttt{#1}}}}

\newcommand{\MS}[1]{{\leavevmode\color{Periwinkle}{#1}}}
\newcommand{\GM}[1]{{\leavevmode\color{Emerald}{#1}}}

\newcommand*{\df}{\mathrel{\mathrel{\mathop:}=}}

\newcommand*{\F}{\mathcal{F}}

\newcommand*{\V}{\mathcal{V}}

\newcommand*{\ar}{\mathsf{arity}}

\newcommand*{\asq}{\mathbf{a}}

\newcommand*{\mbsq}{\mathbf{T}}

\newcommand*{\argsq}{\mathbf{H}} 

\newcommand*{\topf}{\mathbf{f}}  

\newcommand*{\wqo}{\leqslant}

\newcommand*{\simpo}{\prec}

\newcommand*{\domein}{\mathrel{\mid}}

\newcommand*{\pr}{\text{'}}

\newcommand*{\grw}{\to_{\mathcal{G}}}

\newcommand*{\lex}{<_{\mathsf{lex}}}

\newcommand*{\lof}{\ll}

\newcommand*{\rd}{\leftarrow}

\newcommand*{\un}{\oplus}

\newcommand*{\lpo}{<_{\mathsf{lpo}}}
\newcommand*{\lpoeq}{\leqslant_{\mathsf{lpo}}}

\newcommand*{\toplex}{\topemb_{\mathsf{lex}}}
\newcommand*{\lexeq}{=_{\mathsf{lex}}}


\newcommand*{\ndes}{\mathcal{N}}

\newcommand*{\succs}{\mathsf{succ}}

\newcommand*{\succi}[1]{\overset{#1}{\rightharpoonup}}

\newcommand*{\lb}{\mathsf{label}}

\newcommand*{\rt}{\mathsf{root}}

\newcommand*{\inlets}{\mathsf{inlets}}

\newcommand*{\sub}{{\upharpoonright}}

\newcommand*{\ctg}{\vartriangle}

\newcommand*{\TG}{\mathcal{TG}}

\newcommand*{\G}{\mathcal{G}}

\newcommand{\Pos}{\mathsf{Pos}}

\newcommand*{\fleq}{\mathrel{\trianglerighteq}}

\newcommand*{\ufleq}{\mathrel{\trianglelefteq}}


\newcommand*{\iso}{\cong}

\newcommand*{\Top}{\mathsf{Top}}
\newcommand*{\Tops}{\mathsf{Tops}}

\newcommand*{\topemb}{\sqsubset}
\newcommand*{\topembeq}{\sqsubseteq}

\newcommand*{\stopemb}{\sqsupset}
\newcommand*{\stopembeq}{\sqsupseteq}

\newcommand*{\ilab}{\ctg}

\newcommand*{\emb}{\sqsubset_{\mathsf{emb}}}
\newcommand*{\embeq}{\sqsubseteq_{\mathsf{emb}}}

\newcommand*{\semb}{\sqsupset_{\mathsf{emb}}}
\newcommand*{\sembeq}{\sqsupseteq_{\mathsf{emb}}}

\newcommand*{\args}{\mathsf{arg}}

\newcommand*{\pembeq}{\embeq^\text{\cite{1997_plump}}}
\newcommand*{\psembeq}{\sembeq^\text{\cite{1997_plump}}}

\newcommand*{\mm}[2]{m(\cial{#1}) = \cial{#2}}


\newcommand*{\ci}[1]{\tikz{ 
  \node[circle, draw, inner sep=1.5pt, minimum size = 9pt] (char) {\tiny{#1}};}}

\newcommand*{\cial}[1]{\tikz [baseline=(char.base)] {
  \node[circle, draw, inner sep=1.5pt, minimum size = 9pt] (char) {\tiny{#1}};}}

\tikzstyle{n}=[rectangle, inner sep=2pt, outer sep=1pt, minimum size=12pt, node distance = 0.2cm and 0.2cm]

\newcommand{\nn}[2]{$#1\mathrel{:}\cial{#2}$}

\tikzstyle{nme}=[node distance = 0.5cm]

\tikzstyle{mo}=[|->, dashed]


\newcommand*{\fun}[1]{\mathsf{#1}}

\newcommand*{\f}{\fun{f}}
\newcommand*{\fa}{\fun{a}}
\newcommand*{\g}{\fun{g}}
\newcommand*{\fb}{\fun{b}}

\newcommand*{\femb}[3]{
    \node [n, xshift= #3 cm]   (1) {$\f$};
    \node [n, xshift=0.5cm, %
           below left = of 1]  (2) {$#1$};
    \node [n, xshift=-0.5cm, %
           below right = of 1] (3) {$#2$};
    \path[->] (1) edge (2)
              (1) edge (3)
    ;
}
\newcommand*{\fembs}[2]{
    \node [n, xshift= #2 cm]   (1) {$\f$};
    \node [n,  %
           below = of 1]  (2) {$#1$};
    \path[->] (1) edge [bend left]  (2)
              (1) edge [bend right] (2)
    ;
} 

\newcommand*{\tct}{\textsf{T\kern-0.2em\raisebox{-0.3em}C\kern-0.2emT}}
\newcommand*{\aprove}{\textsf{AProVE}}
